\documentclass{article}

\usepackage{times}
\usepackage{graphicx}
\usepackage{color}
\usepackage{amssymb}
\usepackage{latexsym}
\usepackage{amsmath}
\usepackage{theorem}
\usepackage{csquotes}

\newtheorem{theorem}{Theorem}
\newtheorem{proposition}[theorem]{Proposition}

\newtheorem{corollary}[theorem]{Corollary}

\newtheorem{definition}[theorem]{Definition}
\newenvironment{proof}{\noindent\textbf{Proof:}\ }{\hfill$\Box$\medskip\par}
{\theorembodyfont{\rmfamily}}
{\theorembodyfont{\rmfamily}}

\title{Polymorphisms and Circuit Complexity}
\author{Gustav Nordh\thanks{E-mail: {\tt gustav.nordh@gmail.com}.
   }
}
\date{\today}

\begin{document}

\maketitle

\begin{abstract}
We present a framework for studying circuit complexity that is inspired by techniques that are used for analyzing the complexity of CSPs. We prove that the circuit complexity of a Boolean function $f$ can be characterized by the partial polymorphisms of $f$'s truth table. Moreover, the non-deterministic circuit complexity of $f$ can be characterized by the polymorphisms of $f$'s truth table. 
\end{abstract}

\section{Introduction}
\label{section:introduction}
It is well known that almost all Boolean functions require circuits of exponential size, but so far we have not been able to pinpoint a single explicit function requiring circuits larger than $5n$.  
The basic idea of our approach is to make use of techniques and results for analyzing the complexity of SAT problems to get a better understanding of circuit complexity.

Let SAT($S$) denote the SAT problem restricted to instances that are conjunctions of constraints build over the relations in $S$. 
The complexity of SAT($S$) is characterized (up to polynomial-time reducibility) by the polymorphisms of $S$, denoted $Pol(S)$~\cite{TCS98}. For now, think about polymorphisms of $S$ as a generalized form of automorphisms, i.e., operations preserving the structure $S$. The richer the polymorphisms of a structure $S$ is, the \enquote{simpler} the structure is. Indeed, SAT($S$) is in P if $S$ has a non-trivial polymorphism, and NP-complete otherwise. It was observed in~\cite{JLNZ03} that the partial polymorphisms of $S$ (i.e., polymorphisms that may be undefined on some inputs), denoted $pPol(S)$, 
paints a more fine grained picture for the complexity of SAT($S$). For example, if $pPol(S) \subseteq pPol(S')$ and SAT($S$) is solvable in $O(c^n)$ time, then SAT($S'$) is solvable in the same time $O(c^n)$ ($n$ denotes the number of variables).

Let $B_{n}$ denote the set of all Boolean functions with $n$ inputs and $1$ output.
Given $f \in B_{n}$ let $f^{\bullet}$ denote the truth table of $f$ i.e., the $(n+1) \times 2^n$ matrix where the first $n$ columns represents the inputs to $f$, the last column represents the output of $f$, and the rows of $f^{\bullet}$ are sorted in lexicographic order.

Our first observation is that if $f^{\bullet}$ is preserved by a non-trivial polymorphism (i.e., one that is not essentially unary or constant), then $f$ has a circuit of size $O(n)$ (Section~\ref{sec:upper}). Hence, it seems that our intuition from SAT($S$) carries over, if the polymorphisms of $f^{\bullet}$ are \enquote{rich}, then $f$ is simple (i.e., has low circuit complexity). To argue in the other direction, namely, that if the polymorphisms of $f^{\bullet}$ are not rich, then $f$ has high circuit complexity, we first give some general intuition.

To compute a Boolean function $f$, we need to avoid all potential errors, i.e., not output a $0$ on some input for which $f$ is $1$ or vice versa. In many computational models (e.g., Turing Machines, NFAs, or circuits), a computation consists of a composition of primitive computation steps/transitions/gates. 

\textbf{Any function $w$ that is not a 
polymorphism of $f^{\bullet}$ represent a potential error. For each such $w$, there must be at least one primitive (step/transition/gate) that catch/cover this error by not being preserved by $w$. Otherwise, $w$ is a polymorphism of the computation, and the error represented by $w$ manifest itself.} 

The smallest number of individual steps/transitions/gates that cover all the potential errors of $f$ (i.e., all $w$ that are not polymorphisms of $f^{\bullet}$) is a lower bound on the complexity of $f$. Thus, the \enquote{poorer} the polymorphisms of $f^{\bullet}$ are, the larger the complexity of $f$ is.

This line of thinking is inspired by the method of approximation that was introduced by Razborov in his celebrated monotone circuit lower bound results~\cite{Razborov85b, Razborov85a} and further extended in~\cite{Razborov89}. The method of approximation was put in a different framework by Karchmer~\cite{Karchmer93} and the method presented in this framework was coined the fusion method by Wigderson in his survey of the topic~\cite{Wigderson93}. For those familiar with this previous line of work, we remark that the notion of a \enquote{fusion functional} (as used by Karchmer and Wigderson) corresponds to the functions $w$ that are not polymorphisms of $f^{\bullet}$. 

Our first main result (Section~\ref{sec:partial}) is that, given $f \in B_n$, the smallest number of gates that cover all (witnesses of) \emph{partial} functions $w$, such that $w\notin pPol(f^{\bullet})$, is exactly the number of gates in an optimal circuit for $f$. Our second result (Section~\ref{sec:total}) is that the smallest number of gates that cover all (witnesses of) \emph{total} functions $w$, such that $w \notin Pol(f^{\bullet})$, equals the \emph{non-deterministic} circuit size of $f$, up to a constant factor.

\section{Preliminaries}
\subsection{Function algebra}
Any operation on $\{0,1\}$ can be extended in a standard way to an operation on tuples 
over $\{0,1\}$, by applying the operation componentwise as follows.
\begin{definition}
Let $w \in B_k$ and let $R$ be an $n$-ary relation over $\{0,1\}$. For any collection of $k$ tuples, 
$t_1,t_2, \dots, t_k \in R$, the $n$-tuple $w(t_1,t_2, \dots ,t_k)$ is defined as follows:
$w(t_1,t_2, \dots ,t_k) = (w(t_1[1],t_2[1], \dots, t_k[1]),$ $w(t_1[2],t_2[2], \dots, t_k[2]),\dots,$ $w(t_1[n],t_2[n], \dots, t_k[n]))$,
where $t_j[i]$ is the $i$th component in tuple $t_j$.
\end{definition}

\begin{definition}
If $w$ is an operation such that for all $t_1,t_2, \dots, t_k \in R$
$w(t_1,t_2, \dots ,t_k) \in R$, then $R$ is closed under $w$. An operation $w$ such that $R$ is closed under $w$ is called a polymorphism of $R$. The set of all polymorphisms of $R$ is denoted $Pol(R)$.
\end{definition}

Consider the following two binary functions $f$ and $g$, which truth tables $f^{\bullet}$ and $g^{\bullet}$ are given below.

\begin{table}[h]
\begin{center}
\begin{tabular}{cc|c}
$x_1$&$x_2$&$f$\\
\hline
$0$&$0$&$0$\\
$0$&$1$&$0$\\
$1$&$0$&$0$\\
$1$&$1$&$1$\\
\end{tabular}
\quad
\begin{tabular}{cc|c}
$x_1$&$x_2$&$g$\\
\hline
$0$&$0$&$1$\\
$0$&$1$&$0$\\
$1$&$0$&$0$\\
$1$&$1$&$1$\\
\end{tabular}
	\caption{Truth tables $f^{\bullet}$ and $g^{\bullet}$}
	\label{tab:ExampleFunctions}
	\end{center}
\end{table}
Note that neither $f^{\bullet}$ nor $g^{\bullet}$ is closed under the ternary majority operation $maj(x,x,y)=maj(x,y,x)=maj(y,x,x) = x$. For $f^{\bullet}$, there is only one \enquote{witness} that $maj$ is not a polymorphism of $f^{\bullet}$, namely applying the $maj$ operation to the last 3 tuples, results in the tuple $(1,1,0)$ which is not in $f^{\bullet}$. The application of $maj$ to any other combination of three tuples in $f^{\bullet}$, results in a tuple in $f^{\bullet}$. If we consider $g^{\bullet}$ instead, applying $maj$ to \emph{any} 3 distinct tuples, results in a tuple which is not in $g^{\bullet}$. Hence, it seems natural to consider $g$ as being further away from being closed under $maj$ than $f$. Further more, every witness of the fact that an operation $w$ is not a polymorphism of the truth table of a function, constitute a potential error that any circuit computing the function must catch. Hence, not only do we need to keep track of operations $w$ that are not polymorphisms of the truth table of the function, but also the set of all witnesses of this (i.e., all combinations of tuples from the truth table, for which applying $w$ results in a tuple which is not in the truth table).

\begin{definition}
For $f \in B_n$, let $\overline{Pol}(f^{\bullet})$ denote the set of all functions $w \in B_{2^n}$ such that $w$ applied to the tuples in $f^{\bullet}$ (sorted in lexicographic order) results in a tuple that is not in $f^{\bullet}$. Hence, each $w \in \overline{Pol}(f^{\bullet})$ represents a witness that some function is not a polymorphism of $f^{\bullet}$. We sometimes refer to functions $w \in \overline{Pol}(f^{\bullet})$ as anti-polymorphisms.
\end{definition}

If we reconsider the truth tables of the functions $f$ and $g$ above, and consider: 
\begin{center}
$w_1(x_1,x_2,x_3,x_4) = maj(x_1,x_2,x_3)$, \\
$w_2(x_1,x_2,x_3,x_4) = maj(x_1,x_2,x_4)$, \\
$w_3(x_1,x_2,x_3,x_4) = maj(x_1,x_3,x_4)$, \\
$w_4(x_1,x_2,x_3,x_4) = maj(x_2,x_3,x_4)$. 
\end{center}
We have $w_1,\dots,w_4 \in \overline{Pol}(g^{\bullet})$ but only $w_4 \in \overline{Pol}(f^{\bullet})$.

Let $P_{n}$ denote the set of all partial Boolean operations/functions with $n$ inputs and $1$ output (i.e., operations that may be undefined for some inputs).
The concept of polymorphisms has a natural extension to partial operations.
\begin{definition}
Let $w \in P_k$ and $R$ an $n$-ary relation, then $R$ is closed under $w$ if for all $t_1,t_2, \dots, t_k \in R$ either $w(t_1,t_2, \dots ,t_k) \in R$ or at least one of $w(t_1[1],t_2[1], \dots, t_k[1]),$ $w(t_1[2],t_2[2], \dots, t_k[2]),\dots,$ $w(t_1[n],t_2[n], \dots, t_k[n])$ is undefined. A $w \in P_k$ such that $R$ is closed under $w$ is called a partial polymorphism of $R$. The set of all partial polymorphisms of $R$ is denoted $pPol(R)$. Note that $Pol(R) \subseteq pPol(R)$.
\end{definition}

\begin{definition}
For $f \in B_n$, let $\overline{pPol}(f^{\bullet})$ denote the set of all functions $w \in P_{2^n}$ such that $w$ applied to the tuples in $f^{\bullet}$ (sorted in lexicographic order) results in a tuple that is not in $f^{\bullet}$. Hence, each $w \in \overline{pPol}(f^{\bullet})$ represents a witness that some (partial) function is not a partial polymorphism of $f^{\bullet}$.
\end{definition}

\subsection{Circuits}
A Boolean circuit is a directed acyclic graph with three types of labeled vertices: sources (in-degree $0$) labeled $x_1,\dots,x_n$, a sink (the output), and vertices with in-degree $k > 0$ are gates labeled by Boolean functions on $k$ inputs. Unless otherwise specified, we assume the gates of the circuit to be fan-in two $\land$ and $\lor$ gates together with $\neg$ gates.
A non-deterministic circuit has, in addition to the ordinary inputs $x = (x_1,\dots,x_n)$, a set of \enquote{non-deterministic} inputs $y=(y_1,\dots,y_m)$. A non-deterministic circuit $C$ accepts input $x$ if there exists $y$ such that the circuit output $1$ on $(x,y)$. A co-non-deterministic circuit $C$ rejects an input $x$ if there exists $y$ such that $C$ output $0$ on $(x,y)$. Let $|C|$ denote the number of gates of a circuit $C$.

A family of non-deterministic circuits $\{C_n\}_{n \geq 0}$, with $C_n$ having $n$ (ordinary) input gates, decide a language $L$ if each $C_n$ decide $L_n$ (i.e., $C_n$ accepts $x$ if and only if $|x| =n$ and $x \in L$). 
The class $NP/poly$ is defined as the class of languages decidable by non-deterministic circuit families $\{C_n\}$, with $|C_n| \leq poly(n)$. Recall that $P/poly$ is the class of languages decidable by (deterministic) circuit families $\{C_n\}$, with $|C_n| \leq poly(n)$.
Similarly, $coNP/poly$ is the class of languages decidable by polynomial size co-non-deterministic circuit families.

\subsection{Covers}
To be able investigate $Pol(f^{\bullet})$ in relation to the circuit complexity of $f \in B_n$ we need to introduce the computational model that we use. We define a $\land$ gate to be any $3 \times 2^n$ Boolean matrix where the third column is the $\land$ of the first two columns. A $\lor$ gate is defined analogously, and a $\neg$ gate is any $2 \times 2^n$ Boolean matrix where the second column is the complement of the first. An input (gate) is any $1 \times 2^n$ matrix that is one of the first $n$ column vectors in $f^{\bullet}$ (i.e, the inputs to $f$). 

Given a gate $g_i$ we denote its input columns (in case they exist) by $g_{i_1}$ and $g_{i_2}$ (a gate may have just one input ($\neg$), or no inputs (i.e., an input gate)), and its output column (the last column in its matrix) by $g_i$ (all gates have an output). 
Denote the function of the gate $g_i$ by $\circ_i$, e.g., $\circ_i \in \{\land,\lor,\neg\}$.
A circuit (or straight line program) is a sequence of gates $P = (g_1,g_2,\dots,g_t)$ (sometimes viewed as a $t \times 2^n$ matrix) such that the first $n$ gates $g_1,\dots,g_n$ are the input gates (i.e., the first $n$ columns of $f^{\bullet}$), and for every $i>n$, the inputs of $g_i$, i.e., 
$g_{i_1}$ and $g_{i_2}$, satisfy $i_1,i_2 < i$. That is, the inputs of $g_i$ must be the outputs of a gate preceding it in the sequence. The computation of $P$ on input $x \in \{0,1\}^n$ is defined as $P(x) = g_1(x) \cdots g_t(x) = u \in \{0,1\}^t$, where $u$ consists of the outputs of all gates in $P$ when propagating the input $x$ through the circuit (i.e., $u$ is the row of the matrix $P = (g_1,g_2,\dots,g_t)$ corresponding to the input $x$). For example, 
$u_i = u_{i1} \circ_i u_{i2}$. We say that $P$ computes $f \in B_{n}$ if $g_t(x) = f(x)$ for all $x \in \{0,1\}^n$. 

The key for obtaining a lower bound is the observation that if some $P$ of length $t$ computes $f$ (i.e., $f$ has a circuit of size $t-n$) then every $w \in \overline{Pol}(f^{\bullet})$ must fail to be consistent with $P$, i.e., 
for some $1 \leq i \leq t$, $w(g_{i_1}) \circ_i w(g_{i_2}) \neq w(g_i)$.
\begin{proposition}
\label{prop:consistent}
If $w \in \overline{Pol}(f^{\bullet})$ is consistent with the program $P = (g_1,g_2,\dots,g_t)$, then $P$ does not compute $f$.
\end{proposition}
\begin{proof}
Assume that $P = (g_1,g_2,\dots,g_t)$ is a circuit computing $f$. Let $x_1,\dots,x_{n+1}$ be the columns of $f^{\bullet}$ and $z = (w(x_1),w(x_2),\dots,w(x_n))$. Since $w \in \overline{Pol}(f^{\bullet})$ we know that $f(z) \neq w(x_{n+1})$. 
Applying $w$ to the columns of the $t \times 2^n$ matrix $P$ results in a $t$-tuple $u = u_1 \dots u_t$ possibly corresponding to a correct computation of $z$ by the circuit. 
We know that $w$ is consistent with $P$, i.e., 
$w(g_{i_1}) \circ_i w(g_{i_2}) = w(g_i)$
for all gates $g_i$, and hence, 
$u_i = u_{i_1} \circ_i u_{i_2}$, and $u$ represents a correct computation of $P$ on input $z$. This leads to a contradiction since $f(z) = u_t \neq w(x_{n+1}) = u_t$.
\end{proof}
A gate $g_i$ is said to cover $w \in \overline{Pol}(f^{\bullet})$ if $w(g_{i_1}) \circ_i w(g_{i_2}) \neq w(g_i)$
\begin{definition}
\label{def:cover2}
A collection of gates $\mathcal{T}$ that cover all $w \in \overline{Pol}(f^{\bullet})$ 
is said to be a $Pol$ cover for $f$. A minimal $Pol$ cover (in terms of number of gates) for $f$ is denoted $\mathcal{T}(f)$, and its size (i.e., the number of gates) is denoted $|\mathcal{T}(f)|$.
\end{definition}
A cover for $f$ can be seen as an (unsorted) collection of gates that together catch all the potential errors that a circuit for $f$ must deal with. The idea is that since a cover is a simpler object than a circuit, it might be easier to prove lower bounds on the size of a cover for $f$ than the size of a circuit for $f$.

A gate $g_i$ is said to cover $w \in \overline{pPol}(f^{\bullet})$ if $w(g_{i_1})$ and $w(g_{i_2})$ are defined and 
$w(g_{i_1}) \circ_i w(g_{i_2}) \neq w(g_i)$
(i.e, if $w$ is defined on the inputs to $g_i$ but $w$ is undefined on the output of $g_i$ or not consistent with $g_i$).
A collection of gates $\mathcal{P}$ cover $\overline{pPol}(f^{\bullet})$ if each $w \in \overline{pPol}(f^{\bullet})$ is covered by at least one gate $g_i \in \mathcal{P}$.
\begin{definition}
\label{def:cover}
A collection of gates $\mathcal{P}$ that cover $\overline{pPol}(f^{\bullet})$ such that: (1) no two gates in $\mathcal{P}$ output the same result, (2) the result column of $f^{\bullet}$ is not an input of any gate in $\mathcal{P}$, and (3) none of the input columns of $f^{\bullet}$ is an output of a gate in $\mathcal{P}$, is said to be a $pPol$ cover for $f$. A minimal $pPol$ cover (in terms of number of gates) for $f$ is denoted $\mathcal{P}(f)$, and its size (i.e., the number of gates) is denoted $|\mathcal{P}(f)|$.
\end{definition}
We remark that conditions (1)-(3) in Definition~\ref{def:cover} are used to avoid cycles when converting a $pPol$ cover to a circuit later on. They can be replaced by requiring $pPol$ covers to be acyclic. 

\section{Non-trivial polymorphisms implies trivial circuits}
\label{sec:upper}
In this section we note that if $Pol(f^{\bullet})$ contains a non-trivial polymorphism, then the circuit complexity of $f \in B_{n}$ is at most $O(n)$. By a non-trivial polymorphism we mean any polymorphism which is not a constant function, a projection, or the negation of a projection.

\begin{theorem}
Given $f \in B_{n}$, if $Pol(f^{\bullet})$ contains a non-trivial polymorphism, then $f$ has a circuit of size $O(n)$.
\end{theorem}
\begin{proof}
By inspection of Post's lattice of Boolean clones~\cite{post}, we know that if $Pol(f^{\bullet})$ contains a non-trivial polymorphism, then it must contain at least one of the following four polymorphisms: 
\begin{enumerate} 
\item the majority operation $maj(x,x,y)=maj(x,y,x)=maj(y,x,x) = x$ 
\item the affine operation $aff(x,y,z) = x \oplus y \oplus z$ (where $\oplus$ is addition modulo $2$)
\item the and operation $and(x,y) = x \land y$
\item the or operation $or(x,y) = x \lor y$
\end{enumerate}

Given $f \in B_{n}$, in order to design our circuit $C$ we first pre-compute $f$ on the all $0$ input, the all $1$ input, the $n$ inputs having exactly one $1$, and the $n$ inputs having exactly one $0$. More formally, $t_{i}$ ($1 \leq i \leq n$) is the output of $f$ on the input that has a unique $1$ in position $i$, $t_{n+i}$ ($1 \leq i \leq n$) is the output of $f$ on the input that has a unique $0$ in position $i$, $t_{2n+1}$ is the output of $f(0,0,\dots,0)$, and $t_{2n+2}$ is the output of $f(1,1,\dots,1)$. We hard-wire these $(2n+2)$ bits of information $t_{i}$, $1 \leq i \leq 2n+2$ in our circuit. 

The task of the circuit $C$ on input $x = (x_1,x_2,\dots,x_n)$ is to repeatedly apply the non-trivial polymorphism of $f^{\bullet}$ to these $(2n+2)$ bits $t_{i}$ until we arrive at the output of $f(x)$.

In the case where $Pol(f^{\bullet})$ contains the $or$ operation: \\
01: $r := t_{2n+1}$;\\
02: for $1 \leq i \leq n$ \{ \\  
03: \quad if $x_i = 1$ \{ \\
04: \quad \quad $r := r \lor t_{i}$; \\
05: \quad \} \\
06: \} \\
07: return $r$; \\

In the case where $Pol(f^{\bullet})$ contains the $and$ operation: \\
01: $r := t_{2n+2}$;\\
02: for $1 \leq i \leq n$ \{ \\  
03: \quad if $x_i = 0$ \{ \\
04: \quad \quad $r := r \land t_{n+i}$; \\
05: \quad \} \\
06: \} \\
07: return $r$; \\

In the case where $Pol(f^{\bullet})$ contains the $aff$ operation: \\
01: $r := t_{2n+1}$;\\
02: for $1 \leq i \leq n$ \{ \\  
03: \quad if $x_i = 1$ \{ \\
04: \quad \quad $r := t_{2n+1} \oplus r \oplus t_{i}$; \\
05: \quad \} \\
06: \} \\
07: return $r$; \\

In the case where $Pol(f^{\bullet})$ contains the $maj$ operation: \\
01: $r := t_{2n+2}$;\\
02: for $1 \leq i \leq n$ \{ \\  
03: \quad if $x_i = 0$ \{ \\
04: \quad \quad $r := maj(t_{2n+1}, r, t_{n+i})$; \\
05: \quad \} \\
06: \} \\
07: return $r$; \\

To see that the circuit $C$ on input $x = (x_1,x_2,\dots,x_n)$ output $r = f(x)$, we consider the case 
where $or \in Pol(f^{\bullet})$ (the arguments in the other cases are very similar). In line 01 we initialize $r$
to be the output of $f(0,0,\dots,0)$. Then (in lines 02-04) we take the $\lor$ of all $f(0,\dots,0,x_i,0,\dots,0)$ for which $x_i = 1$. This is the final output $r$. The fact that $r = f(x)$ follows from $or \in Pol(f^{\bullet})$ since if we take the $\lor$ of all the inputs $(0,\dots,0,x_i,0,\dots,0)$ for which $x_i = 1$ (i.e., all the inputs corresponding to the outputs we took $\lor$ of), we arrive at the original input vector $x = (x_1,x_2,\dots,x_n)$.
Note that the circuit $C$ has size $O(n)$ as the number of bits that we hard wire is $O(n)$, and in each of the $n$ iterations of the for loop we carry out a constant number of operations.
\end{proof}

Note that the construction above is easy to extend to multi-output functions. Given a Boolean function $f$ with $n$ inputs and $m$ outputs, such that $Pol(f^{\bullet})$ (where $f^{\bullet}$ is now a $(n+m) \times 2^n$ matrix) contains a non-trivial polymorphism, the construction results in a circuit of size $O(nm)$.

Also note that it is easy to extend this upper bound to functions $f \in B_{n}$ for which $Pol(f^{\bullet})$ is \enquote{close} to contain a non-trivial polymorphism. For example, if we can modify at most $n^k$ outputs of $f$ such that the truth table of the resulting function $g^{\bullet}$ is closed under a non-trivial polymorphism, then $f$ has circuits of size $O(n^k)$. This is because we can hard wire in our circuit $C$ the correct outputs corresponding to the outputs that were modified. The circuit $C$ is then designed as before for computing $g$ instead. On input $x$ the circuit first checks whether $x$ corresponds to a modified output, and if this is the case, it looks up the correct output $f(x)$. Otherwise $g(x)$ is computed, as before.

\section{Partial polymorphisms and deterministic circuits}
\label{sec:partial}
In this section we prove that the circuit complexity of $f$ can be characterized by the partial polymorphisms of $f^{\bullet}$. More precisely, we prove that a collection of gates is a minimal $pPol$ cover for $f$ if and only if the collection of gates form an optimal circuit for $f$.
\begin{proposition}
$|\mathcal{P}(f)|$ is a lower bound on the circuit complexity of $f$.
\end{proposition}
\begin{proof}
Let $\mathcal{P}(f)$ be the gates in an optimal circuit $C$ for $f$. Note that conditions (1)-(3) in the definition of a $pPol$ cover (Definition~\ref{def:cover}) are satisfied by the gates of any optimal circuit for $f$. Assume there is a $w \in \overline{pPol}(f^{\bullet})$ that is not covered by $\mathcal{P}(f)$. Hence, for every gate $g$ in $\mathcal{P}(f)$, $w$ is either undefined on an input to $g$ or $w$ is consistent with $g$. If $w$ is undefined on an input to $g$, then $w$ must be undefined on an output of a direct predecessor $g'$ to $g$ (since $C$ is a circuit). Without loss of generality assume that $w$ is defined for all inputs to gates that precedes $g$ in $C$. Hence, $w$ is covered by $g'$ (contradicting that $w$ is not covered by $\mathcal{P}(f)$). Thus, $w$ must be defined on, and consistent with, all the gates in $\mathcal{P}(f)$. By the same reasoning as in Proposition~\ref{prop:consistent}, this is impossible due to $w \in \overline{pPol}(f^{\bullet})$, and we conclude that $w$ is covered by $\mathcal{P}(f)$.
\end{proof}

\begin{proposition}
$|\mathcal{P}(f)|$ is an upper bound on the circuit complexity of $f$.
\end{proposition}
\begin{proof}
Given an optimal cover $\mathcal{P}(f)$, unless $f$ is a projection of one of its inputs (in which case $\mathcal{P}(f)$ is empty), we note that 
the result column of $f^{\bullet}$ (i.e., the last column $f^{\bullet}$) which we denote $r$, must be a column of one of the gates in $\mathcal{P}(f)$. If not, consider $w \in P_{2^n}$ that is a projection on its $i$th coordinate for all inputs except $r$, for which $w$ is the negation of its $i$th coordinate. Thus, $w \in \overline{pPol}(f^{\bullet})$, and $w$ is consistent with all gates in $\mathcal{P}(f)$, which is a contradiction.

Assume there is an input to a gate $g \in \mathcal{P}(f)$ that is not an input to $f$ and that is not an output of a gate in $\mathcal{P}(f)$. Since $\mathcal{P}(f)$ is minimal there is a $w \in \overline{pPol}(f^{\bullet})$ that is covered only by $g$ (and no other gate in $\mathcal{P}(f)$). Let $w'$ be undefined on the input to $g$ assumed above, but otherwise identical to $w$. Hence, $w'$ is not covered by $\mathcal{P}(f)$ and $w' \in \overline{pPol}(f^{\bullet})$, which is a contradiction
with the fact that $\mathcal{P}(f)$ is a cover. 

Thus, the gates in $\mathcal{P}(f)$ would form an optimal circuit computing $f$, should no cycles be present. Utilizing conditions (1)-(3) in the definition of a $pPol$ cover (Definition~\ref{def:cover}), we can show that cycles are impossible. Assume that $\mathcal{P}(f)$ contains a cycle and pick an arbitrary gate $g$ on the cycle. Again, since $\mathcal{P}(f)$ is minimal there is a $w \in \overline{pPol}(f^{\bullet})$ that is covered only by $g$. Let $w'$ be identical to $w$ except that $w'$ is undefined on every output (and hence at least one input) of all the gates in the cycle that $g$ belongs to. By conditions (2) and (3), none of the columns in $f^{\bullet}$ can be part of a cycle, and hence $w' \in \overline{pPol}(f^{\bullet})$. To see that $w'$ is not covered by $\mathcal{P}(f)$, note that $g$ does not cover $w'$ as $w'$ is undefined on an input to $g$. If another gate $g'$ in $\mathcal{P}(f)$ cover $w'$ it must be because $w'$ (as opposed to $w$) is undefined on the output of $g'$, implying that $g'$ has the same output as a gate on the cycle, which is impossible by condition (1).
\end{proof}

\begin{corollary}
$P/poly$ is the class of languages defined by functions having polynomial $pPol$ covers, i.e., $\{f_n \in B_n\}_{n \geq 0}$ with $|\mathcal{P}(f_n)| \leq poly(n)$.
\end{corollary}

\section{Polymorphisms and non-deterministic circuits}
\label{sec:total}
We first introduce a special type of non-deterministic circuits called total single-valued non-deterministic circuits (TSVND circuits). These circuits have appeared previously in the literature mainly in relation to derandomization of Arthur-Merlin games, see for example~\cite{GST03,MV05}.
\begin{definition}\cite{GST03}
A TSVND circuit is a non-deterministic circuit $C(x,y)$ with three possible outputs $0,1$ and $quit$, such that for each $x \in \{0,1\}^n$, either $\forall y C(x,y) \in \{0,quit\}$ or $\forall y C(x,y) \in \{1,quit\}$. That is, there can be no $y,y'$ such that $C(x,y) = 1$ and $C(x,y') = 0$, and we define $C(x) = b \in \{0,1\}$ if there exist $y$ such that $C(x,y) = b$, and $C(x) = quit$ if there is no such $y$. Finally, we require $C$ to define a total function on $\{0,1\}^n$, i.e., for each $x \in \{0,1\}^n$ $C(x) \neq quit$. 
\end{definition}

The following fact about TSVND circuit complexity is easy to realize.
\begin{proposition}
$f \in B_n$ has TSVND circuit complexity $O(s(n))$ if and only if $f$ has non-deterministic circuit complexity $O(s(n))$ and co-non-deterministic circuit complexity $O(s(n))$.
\end{proposition}
\begin{proof}
Given a non-deterministic circuit $C_1$ for $f$ (with non-deterministic inputs $y_1$) and a co-non-deterministic circuit $C_2$ for $f$ (with non-deterministic inputs $y_2$) we construct a TSVND circuit $C$ for $f$ (with non-deterministic inputs $y_1,y_2$) by using $C_1$ and $C_2$ as sub circuits. Let $C(x,y_1,y_2) = 1$ if $C_1(x,y_1) = 1$, $C(x,y_1,y_2) = 0$ if $C_2(x,y_2) = 0$, and $C(x,y_1,y_2) = quit$ otherwise.

Given a TSVND circuit $C(x,y)$ for $f$ we construct a non-deterministic circuit $C_1(x,y)$ for $f$ by changing all $quit$ outputs in $C$ to $0$. Similarly, we construct a co-non-deterministic circuit $C_2(x,y)$ for $f$ by changing all $quit$ outputs in $C$ to $1$.
\end{proof}

\begin{proposition}
$|\mathcal{T}(f)|$ is a lower bound on the TSVND circuit complexity of $f$.
\end{proposition}
\begin{proof}
Given an optimal TSVND circuit $C(x,y)$ for $f$ (with $n$ (ordinary) inputs $x$ and $m$ non-deterministic inputs $y$) we construct a cover $\mathcal{T}$, by for each $x$ fixing a witness $y$ such that $C(x,y) = b \in \{0,1\}$.
Denote by $fY$ the $(n+m)\times 2^n$ matrix resulting from appending to each input $x$ the corresponding witness $y$ and sorting the rows in lexicographic order. Each gate $g$ of $C$ (which is a $2^{n+m}$-gate) is transformed into $2^n$-gate $g'$.
For each $1 \leq i \leq 2^n$ the $i$th row of $g'$ is the input(s) and output of $g$ when $C$ is passed the inputs $(x,y)$ where $(x,y)$ is the $i$th row of $fY$. Let $\mathcal{T}$ denote the resulting collection of $2^n$-gates. Note that the number of $2^n$-gates in $\mathcal{T}$ is $|C|$. 

To prove that $\mathcal{T}$ is a $Pol$ cover for $f$, assume to the contrary that there is some $w \in \overline{Pol}(f^{\bullet})$ that is not covered by any gate in $\mathcal{T}$. 
Order the gates of $C$ such that no gate has an output which is the input of a gate earlier in the order, with the last gate in the order being the output gate. Order the result columns of all the gates in $\mathcal{T}$ in the exact same order and append them to $fY$. Denote the resulting $(n+m+|C|)\times 2^n$ matrix by $fY\mathcal{T}$ and let $v$ be the vector (of length $(n+m+|C|)$) resulting from applying $w$ to the columns of $fY\mathcal{T}$. 

We claim that $v$ represents a correct computation of $(w(x),w(y))$ in $C$. If not, then there is a gate $g_i$ 
with inputs $v_{i_1}$ and $v_{i_2}$ such that $v_{i_1} \circ_i v_{i_2} \neq v_i$ ($v_i$ is the output of $g_i$). But this is impossible since applying $w$ to the gate $g'_i$ (i.e., the gate in $\mathcal{T}$ corresponding to $g_i$) results in $(v_{i_1}, v_{i_2}, v_i)$, 
and $v_{i_1} \circ_i v_{i_2} = v_i$ since $w$ by assumption is not covered by any gate in $\mathcal{T}$. Hence, $v$ represents a correct computation $(w(x),w(y))$ in $C$. Note that the last element of $v$ (i.e., $v_r$ with $r=n+m+|C|$) is $f(w(x))$ since the last column of $fY\mathcal{T}$ is the result of the output gate. Thus, $C$ outputs $f(w(x))$ on input $(w(x),w(y))$, i.e., $C(w(x)) = f(w(x))$.

By the assumption that $w \in \overline{Pol}(f^{\bullet})$, we have $f(w(x)) \neq v_r$, and by the reasoning above we have $v_r = f(w(x))$.
Hence, $\mathcal{T}$ is a $Pol$ cover for $f$, and $|\mathcal{T}(f)|$ is a lower bound on the TSVND circuit complexity of $f$.
\end{proof}

\begin{proposition}
$f$ has TSVND circuit complexity $O(|\mathcal{T}(f)|)$.
\end{proposition}
\begin{proof}
Given a $Pol$ cover $\mathcal{T}$ for $f$, we show how to construct a TSVND circuit $C$ for $f$ of size $O(|\mathcal{T}|)$.
First note that the result column of $f^{\bullet}$ (i.e., the last column of $f^{\bullet}$), which we denote by $r$, must be a column of one of the matrices in $\mathcal{T}(f)$. If not, consider $w \in B_{2^n}$ that is a projection on its $i$th coordinate for all inputs except $r$, for which $w$ is the negation of its $i$th coordinate. Thus, $w \in \overline{Pol}(f^{\bullet})$, and $w$ is consistent with all gates in $\mathcal{T}$, which is a contradiction with the definition of a cover.

Name the columns of $f^{\bullet}$ $x_1,\dots,x_n,x_{n+1}$ (note that $x_{n+1} = r$). Name each column of $\mathcal{T}$ by the corresponding $x_i$, in case it appears in $f^{\bullet}$, otherwise name it $y_i$ such that identical columns get the same name and no two different columns get the same name. 
The $x_i$'s are the deterministic inputs to $C$ and the $y_i$'s are the non-deterministic inputs. We hard code each gate from $\mathcal{T}$ in the circuit $C$ with the names given, i.e., if the gate is a $\land$ gate with the first two columns being $x_2$, $y_5$ and the last being $x_1$, we store it as $x_2 \land y_5 = x_1$.

On input $(x,y) = (x_1,\dots,x_n,x_{n+1},y_1,\dots,y_m)$, $C$ outputs $quit$ if $(x,y)$ is not a consistent assignment to the variables in the stored gates, and $x_{n+1} = r$ otherwise. First note that for all $x$ there is an $y$ such that $C(x,y) \in \{0,1\}$, namely, let $y$ be the assignment resulting from taking the row identified by $x$ in $\mathcal{T}(f)$. Secondly, for each $x$ there can be no $y$ and $y'$ such that
$C(x,y) = 1$ and $C(x,y') = 0$, since then one of $(x,y)$ or $(x,y')$ would correspond to a $w \in \overline{Pol}(f^{\bullet})$ that is not covered by $\mathcal{T}$. Hence, $C$ is a TSVND circuit computing $f$.
  
Note that as the amount of information that we need to hard code in $C$ is at most a constant times $|\mathcal{T}|$, and the operation of the circuit is a simple evaluation, $C$ has size $O(|\mathcal{T}|)$.
\end{proof}

\begin{corollary}
$NP/poly \cap coNP/poly$ is the class of languages defined by functions having polynomial $Pol$ covers, i.e., $\{f_n \in B_n\}_{n \geq 0}$ with $|\mathcal{T}(f_n)| \leq poly(n)$.
\end{corollary}


\bibliographystyle{abbrv}
\bibliography{references}

\end{document}